\documentclass[12pt,reqno]{amsart}
\usepackage{a4}
\usepackage{enumerate}
\usepackage{amssymb}
\usepackage{amsmath}
\usepackage{amsthm}
\usepackage{pstricks}
\usepackage{fp}

\newcommand{\R}{\mathbb{R}} 
 \newcommand{\N}{\mathbb{N}}
\newcommand{\NZ}{\N_0} 

\newcommand{\lk}{\langle}
\newcommand{\rk}{\rangle}
\newcommand{\scp}[2]{\lk {#1},{#2}\rk}
\newcommand{\norm}[1]{{\| {#1} \|}}
\newcommand{\normsqr}[1]{\norm{#1}{}^2}
\newcommand{\seb}[1]{\operatorname{SEB}(#1)}

\newcommand{\conv}[1]{\operatorname{conv}(#1)}
\newcommand{\argmin}{\operatorname{argmin}}
\newcommand{\uast}{u^\ast}
\newcommand{\uastast}[2]{u^{\ast\ast}_{#1,#2}}

\newcommand{\sxg}{{\tilde X}}
\newcommand{\sxs}{X}
\newcommand{\syg}{{\tilde Y}}
\newcommand{\sys}{Y}

\newcommand{\cartesian}[1]{(#1)}
\newcommand{\polar}[2]{[#1;#2]}

\newcommand{\uu}{u}
\newcommand{\xx}{x}
\newcommand{\cc}{\chi}
\newcommand{\dd}{\xi}
\newcommand{\eps}{\epsilon}

\newcommand{\xw}{\phi}
\newcommand{\cw}{\psi}
\newcommand{\uw}{\alpha}
\newcommand{\la}{\lambda}
\newcommand{\fee}{\phi}
\newcommand{\feebar}{\bar\phi}
\newcommand{\feetwobar}{\bar{\bar\phi}}
\newcommand{\transl}{\mathcal{T}}
\newcommand{\bv}[1]{e_{#1}}
\newcommand{\uset}{V} 
\newcommand{\forbiddenset}{T}

\newtheorem{theorem}{Theorem}
\newtheorem{corollary}[theorem]{Corollary}
\newtheorem{lemma}[theorem]{Lemma}
\newtheorem{proposition}[theorem]{Proposition}
\theoremstyle{definition}
\newtheorem{example}[theorem]{Example}

\begin{document}
\author{Thomas Binder}
\address{University of L{\"u}beck,
Institute of Mathematics, Wallstra{\ss}e 40,
D-23560 L{\"u}beck, Germany}
\email{thmsbinder@gmail.com}
\author{Thomas Martinetz}
\address{University of L{\"u}beck,
Institute for Neuro- und Bioinformatics,
Ratzeburger Allee 160,
D-23538 L{\"u}beck, Germany}
\email{martinetz@informatik.uni-luebeck.de}

\title[On the boundedness of an iteration]{%
On the boundedness of an iteration involving
points on the hypersphere}

\begin{abstract}
For a finite set of points $\sxs$ on the unit hypersphere in $\R^d$
we consider the iteration $\uu_{i+1}=\uu_i+\cc_i$, where
$\cc_i$ is the point of $\sxs$ farthest from $\uu_i$.
Restricting to the case where the origin is contained in
the convex hull of $\sxs$ we study the maximal length of $\uu_i$.
We give sharp upper bounds for the length of $\uu_i$ independently
of $\sxs$. Precisely, this upper bound is infinity
for $d\ge 3$ and $\sqrt2$ for $d=2$.
\end{abstract}

\maketitle

\nocite{FG03}

\section{Introduction and overview}

Throughout this paper we will assume that $d\ge 2$. By
$\R^d$ we denote $d$-dimensional Euclidean space,
equipped with the standard scalar product $\scp{\cdot}{\cdot}$
and induced norm $||\cdot||$.
Moreover $S^l(r)$ denotes the $l$-dimensional sphere of radius $r$,
and $S^l:=S^l(1)$.
These spheres are always considered as embedded in $\R^d$.
Let $\sxs=\{\xx_1,\dots,\xx_n\}\subseteq S^{d-1}\subseteq \R^d$
be a finite set on the unit hypersphere.
Without mentioning this each time, we assume that
the linear space spanned by the elements of $\sxs$ equals $\R^d$,
i.e.\ $d$ cannot be reduced. Consider the iteration
\[
   \uu_0:=0,\qquad \uu_{i+1} := \uu_i + \cc_i,
\]
where $i\in\NZ$ and
$\cc_i$ is the element of $\sxs$ which is farthest away from
$\uu_i$ (which happens to be $\argmin_{x\in \sxs} \scp{x}{\uu_i}$).
In case there are several elements of
$\sxs$ at maximal distance, just choose any of them.
Due to this ambiguity there are many iterations
$(\uu_i)_{i=0}^\infty$ for a particular set $\sxs$.
By $U(\sxs)$ we denote the set of vectors occurring in any of these iterations.
Let
\[
  \uast(\sxs) := \sup\,\{\,\norm{u} \,|\, u\in U(\sxs)\}
\]
be the greatest length reached during any of these iterations.
The question which values $\uast(\sxs)$ can take is simple and intriguing;
it was brought up in connection with the rate of convergence of an
iterative approach of computing the smallest enclosing ball of a
point set, as described in the following.


Let $\syg\subseteq\R^d$ be a finite set of points.
Then the smallest enclosing ball $\seb{\syg}$ of $\syg$
exists and is unique \cite{Wel91}.
We assume that $\syg$ has at least two elements.
By $c\in\R^d$ and $R\in\R^+$ we denote center and radius of $\seb{\syg}$,
respectively.
B\u{a}doiu and Clarkson \cite{BC03} introduced
the following approximation of $c$:
\begin{equation}\label{eq:approx}
  c_0 := 0,\qquad
  c_{i+1} := c_i + \frac{1}{i+1}(\dd_i-c_i),
\end{equation}
where $i\in\N$ and $\dd_i$ is the element of $\syg$ farthest away from $c_i$.
This approximation $(c_i)_{i=0}^\infty$ is related to the iteration
$(\uu_i)_{i=0}^\infty$ by
$R\uu_i=i(c_i-c)$ which implies $\uu_{i+1}=\uu_i+\frac{\dd_i-c}{R}$.
The set $\sxg$ connected to $(\uu_i)_{i=0}^\infty$ is given by
\begin{equation}\label{eqn:xfromy}
\sxg := \Bigl\{ \frac{1}{R} (y - c) \,\,\bigl|\bigr.\,\, y\in\syg \Bigr\}. 
\end{equation}
Unlike $\sxs$ the set $\sxg$ can contain also points in the interior
of the unit hypersphere.
Martinetz, Madany and Mota \cite{MMM06} show that after a finite
number of steps all $\dd_i$ will lie on the boundary of $\seb{\syg}$,
i.e.\ $\dd_i\in\sys$ for all $i\ge i_0$, where $\sys\subseteq\syg$
consists of all points on the surface of $\seb{\syg}$.
This clarifies the correspondence.

While the approximation is extremely easy to use, the question
of convergence needs to be answered.
In \cite{BC03} it is shown that for $i\in\N$
\begin{equation}\label{eq:convspeed}
  \frac{\norm{c-c_i}}{R} \le \frac{1}{\sqrt{i}}.
\end{equation}
\cite{MMM06} aims at proving faster convergence
than \eqref{eq:convspeed}. In particular:

\begin{theorem}[\cite{MMM06}, Theorem~2]\label{thm:mmm2}
Let $\syg\subseteq\R^d$ be a finite set with at least two elements,
and let $\sxg$ be given by \eqref{eqn:xfromy}.
Consider the approximation \eqref{eq:approx} of $\seb{\syg}$.
Then for all $i\in\N$
\[
   \frac{\norm{c - c_i}}{R} \le \frac{\uast(\sxg)}{i},
\]
where the definition of $\uast$ has been extended
to sets $\sxg$ with points on or in the interior of the unit
hypersphere in a straightforward manner.
\end{theorem}
In view of Theorem~\ref{thm:mmm2}, a finite value of $\uast$ or
even a uniform upper bound independent of $\sxs$ is desirable.
Before stating our results on the latter, we need some preparations.

The connection between $(c_i)_{i=0}^\infty$ and $(\uu_i)_{i=0}^\infty$
is further illustrated by
\begin{proposition}\label{prop:zerobal}
For a finite set $\sxs\subseteq S^{d-1}\subseteq \R^d$
the following statements are equivalent.
\begin{enumerate}
\item $\seb{\sxs} = S^{d-1}$,
\item The origin $0\in\R^d$ is contained in $\conv{\sxs}$,
\item $\delta(\sxs)\ge 0$, where
\[
   \delta(\sxs) := - \max_{\norm{u}=1} \min_{x\in \sxs} \scp{x}{u}.
\]
\end{enumerate}
\end{proposition}
\begin{proof}
(i)$\Longleftrightarrow$(ii) is due to R.~Seidel
(cf.\ Lemma~1 in \cite{FGK03}).
(ii)$\Longleftrightarrow$(iii) follows from the fact that
a point $p\in\R^d$ lies in the convex hull of $\sxs$ if and
only if $\min_{x\in \sxs} \scp{x-p}{u} \le 0$ for all unit vectors $u$.
\end{proof}

$\sxs$ is called $0$-balanced if $0\not\in\conv{\sxs}$.
For $1\le b\le d-1$ the set
$\sxs$ is called $b$-balanced, if $0$ is a point
on the boundary of $\conv{\sxs}$ and is contained in a
$b$-dimensional face, but not in a $(b-1)$-dimensional
face of $\conv{\sxs}$.
If $0$ is an inner point of $\conv{\sxs}$, then $\sxs$ is
called $d$-balanced or balanced.
Having the same balance property is an
equivalence relation on all sets $\sxs$ under consideration.

Note that $\delta(\sxs)$ is strictly positive if and only if
$\sxs$ is $d$-balanced, and Proposition~\ref{prop:zerobal} characterizes
all sets $\sxs$ that are not $0$-balanced.

\begin{theorem}\label{thm:mmm}
Let $\sxs$ be a finite set of unit vectors in $\R^d$.
\begin{enumerate}
\item If $\sxs$ is $0$-balanced, then $\uast(\sxs)=\infty$.
\item If $\sxs$ is $b$-balanced for $0<b\le d$,
then $\uast(\sxs)<\infty$.
\end{enumerate}
\end{theorem}

\begin{proof}
Again, (ii) is shown in \cite{MMM06}; it remains to prove (i).
Since $\conv \sxs$ is compact,
there is a point $T\in\conv \sxs$ which is closest to the
origin. Let $\epsilon:=|OT|$. Clearly $||\cc_j||\ge\epsilon$ for
all $j\in\NZ$, therefore
$||u_i||=||\sum_{j=0}^{i-1} \cc_j|| \ge i\epsilon$
is an unbounded sequence for $i\in\NZ$.
\end{proof}

For $0\le b\le d$ we define
\[
  \uastast{d}{b} := \sup\,\{\,\uast(\sxs) \,|\,
    \sxs\subseteq S^{d-1}\subseteq \R^d\,\,
    \text{finite and $b$-balanced} \}.
\]
Our goal is to compute $\uastast{d}{b}$ for all possible $d$ and $b$.

\begin{theorem}\label{thm:main2}
For $d=2$ we have $\uastast{2}{0}=\infty$,
while $\uastast{2}{1}=\uastast{2}{2}=\sqrt{2}$.
\end{theorem}
Clearly, for $d=2$, $\sxs=\{\xx_1,\xx_2\}$, $\xx_1=(0,1)$,
$\xx_2=(1,0)$ the iteration $\uu_0=0$, $\uu_1=\xx_1$, $\uu_2=\xx_1+\xx_2$
is valid and $\norm{\uu_2}=\sqrt2$. This manifest example represents
one inequality of the proof of Theorem~\ref{thm:main2}; the missing
inequality is shown in Section~\ref{section:proof2}.

\begin{theorem}\label{thm:main3}
For $d\ge 3$ we have $\uastast{d}{b}=\infty$
for all $0\le b\le d$.
\end{theorem}

\begin{proof}
For any dimension $d$ we have
$\uastast{d}{0}=\infty$ from Theorem~\ref{thm:mmm}~(i).
For $1\le b\le d-2$ the assertion follows
from the example discussed in
Proposition~\ref{proposition:unbounded} below.
For $b=d$ and $b=d-1$ use
Proposition~\ref{proposition:unbounded2}~(ii) and (iii),
respectively.
\end{proof}

Although the balance property of $\sxs$ is a suggesting
geometric property, it does not seem to give a finer prediction for
$\uast(\sxs)$ than $\delta(\sxs)$. In the balanced case,
$0<\delta(\sxs)$ determines a finite upper bound for
$\uast(\sxs)$ as shown in \cite{MMM06}, namely
\[
  \norm{\uu_i} \le \frac{1}{2\delta(\sxs)} + 1,\qquad i\in\NZ.
\]
With respect to the faster convergence we have
an immediate result for $d=2$:
\begin{corollary}
Let $\syg\subseteq\R^2$ be a finite set with at least two elements.
Assume that all elements of $\syg$ lie on the boundary of
$\seb{\syg}$. Then
$\norm{c-c_i} \le \frac{\sqrt2 R}{i}$ for all $i\in\N$.
\end{corollary}

\section{Proof for $d=2$}
\label{section:proof2}

%
%
Let $\bv1$, $\bv2$ denote the canonical orthonormal basis of $\R^2$.
Each $\xx_j\in X$, $1\le j\le n$ can be written as
\[
 \xx_j = \cos(\xw_j) \, \bv{1} + \sin(\xw_j) \, \bv{2} =
   \polar{1}{\xw_j},
\]
where $\polar{\tilde r}{\tilde\phi}$ indicates a point in standard
polar coordinates on $\R^2$. Similarly, for $j\in\N$ we write
\begin{align*}
 \cc_j &= \cos(\cw_j) \, \bv{1} + \sin(\cw_j) \, \bv{2} =
   \polar{1}{\cw_j},\\
 \uu_j &= \la_j\bigl(\cos(\uw_j) \, \bv{1} + \sin(\uw_j) \, \bv{2}
   \bigr) = \polar{\la_j}{\uw_j}.
\end{align*}
All argument angles are real numbers taken modulo $2\pi$.
The freedom in rotation is fixed as follows.
Assume that $\xx_1,\dots,\xx_n$ are numbered
counterclockwise, starting at $\xw_1=2\pi-\fee$,
ending at $\xw_n=\pi+\fee$, such that there is a
gap with angle size $\pi-2\fee$
between the two neighboring elements $\xx_1$, $\xx_n$
of $X$ is symmetric about the $\bv{2}$-axis.
We call this a parametrization of $X$ with base gap
of size $\pi-2\fee$, where $\fee\in[0,\frac\pi2)$.
The choice of $\fee$ indicates that we restrict to the balanced
cases. Define $\feebar:=\frac{\pi}{6}-\fee$.
For $W\subseteq\R^2$ and $k=1,\dots,n$ let
$\transl_k(W)$ denote the set obtained by translation
of $W$ by $\xx_k$.
The set $\forbiddenset$ is defined by
\[
  \forbiddenset := \Bigl\{
    \polar{\tilde r}{\tilde\phi} \in\R^2 \,\,\bigl|\bigr.\,\,
    \text{$\tilde r\in(1,\sqrt2\,]$ and
    $\tilde\phi\in \Bigl(
      \frac\pi2 - \feebar, \frac\pi2 + \feebar \Bigr)$} \Bigr\}.
\]
Moreover, we define three subsets of $\R^2$ by
\begin{align*}
  R &:= \{ \polar{\tilde r}{\tilde\phi} \,\,|\,\,
     \text{$\tilde r>0$ and
     $\tilde\phi \in (\pi - \fee, 2\pi + \fee)$} \},\\
  Q &:= \{ \cartesian{a,b} \,\,|\,\,
      |a|\tan\fee \le b \le |a|\tan\fee + \lambda_{min}\},\\
  P &:= \{ u\in\R^2 \,|\, \norm{u}\le1 \} \setminus (R \cup Q).
\end{align*}
Here $\lambda_{min} := \frac{\sqrt3}{2\cos\fee}$ is the length
of the intersection of $Q$ with the $\bv{2}$-axis.
Figure~\ref{fig:proof2} gives an illustration of this situation;
\cite{FIG} gives an animated version where $\fee$ varies in time.

\begin{center}
\begin{figure}
%
%
%
%
\newcommand{\dimtwopicture}[2]{

\psset{xunit=#1,yunit=#1,runit=#1}
\begin{pspicture}(-14,-11)(14,15)
\SpecialCoor 
\FProot{\roottwo}{2}{2}
\FProot{\rootthree}{3}{2}
\FPset{\cradius}{10.0}
\FPmul{\cradiusl}{\cradius}{\roottwo}
\FPset{\phee}{#2}
\FPdiv{\pheebytwo}{\phee}{2}
\FPmul{\pheerad}{\phee}{\FPpi}
\FPdiv{\pheerad}{\pheerad}{180}
\FPsub{\pheebar}{30}{\phee}
\FPadd{\pheebarl}{90}{\pheebar}
\FPsub{\pheebarr}{90}{\pheebar}
\FPsub{\xonearg}{360}{\phee}
\FPadd{\xennarg}{180}{\phee}
\FPsub{\xoneopp}{180}{\phee}
\FPadd{\xennopp}{000}{\phee}
\FPcos{\lambdamin}{\pheerad}
\FPdiv{\lambdamin}{\rootthree}{\lambdamin}
\FPmul{\lambdamin}{\lambdamin}{5} 
\FPsin{\lambdamsp}{\pheerad}
\FPmul{\lambdamsp}{\lambdamsp}{\cradius}
\FPsub{\lambdamsp}{\lambdamin}{\lambdamsp}

\FPround\pheebytwo\pheebytwo6
\FPround\pheebarl\pheebarl6
\FPround\pheebarr\pheebarr6
\FPround\xonearg\xonearg6
\FPround\xennarg\xennarg6
\FPround\xoneopp\xoneopp6
\FPround\xennopp\xennopp6
\FPround\lambdamin\lambdamin6
\FPround\lambdamsp\lambdamsp6

\psset{arrows=-,arrowlength=4,arrowinset=0}
%
%
\psline(-14, 0)(14,0)
\psline( 0,-11)(0,15)
\psline[arrows=->](0,0)(\cradius;0)
\psline[arrows=->](0,0)(\cradius;90)
\rput{0}(10.8,-0.7){$\bv{1}$}
\rput{0}( 0.7,10.5){$\bv{2}$}
\pscircle[linecolor=red](0,0){\cradius}
\psline[arrows=->](0,0)(\cradius;\xonearg) 
\rput{0}(11.0;\xonearg){$\xx_1$}
\psline[arrows=->](0,0)(\cradius;\xennarg) 
\rput{0}(11.0;\xennarg){$\xx_n$}
%
\psline[linestyle=dashed](0,0)(14;\xennopp)
\psline[linestyle=dashed,origin={0,-\lambdamin}](0,0)(14;\xennopp)
\psline[linestyle=dashed](0,0)(14;\xoneopp)
\psline[linestyle=dashed,origin={0,-\lambdamin}](0,0)(14;\xoneopp)
%
%
\psline[linestyle=dotted](0,0)(\cradius;\pheebarr)
\psline[linestyle=dotted](0,0)(\cradius;\pheebarl)
%
%
\psline[](\cradius;\pheebarr)(\cradiusl;\pheebarr)
\psline[](\cradius;\pheebarl)(\cradiusl;\pheebarl)
\psarc[](0,0){\cradiusl}{\pheebarr}{\pheebarl}
%
%
\begin{psclip}{
\pspolygon[arrows=-,linestyle=none]%
(\cradius;\pheebarr)(\lambdamin;90)%
(\cradius;\pheebarl)(\cradiusl;90)}
\pswedge[fillstyle=hlines](0,0){\cradius}{\pheebarr}{\pheebarl}
\end{psclip}
%
%
\pswedge[linestyle=none](0,0){\cradius}{\xennopp}{\pheebarr}
\begin{psclip}{\pspolygon[linestyle=none,origin={\cradius;\xennopp}]%
(\cradius;\pheebarr)(\lambdamin;90)(\cradiusl;90)}
\pswedge[fillstyle=hlines](\cradius;\xennarg){\cradius}{\pheebarr}{90}
\end{psclip}
\psline[origin={\cradius;\xennopp}](\cradius;\pheebarr)(\lambdamin;90)
\psline[linestyle=dotted](\cradius;\xennarg)%
(\cradius;\xennarg|0,\lambdamsp)
\psline[linestyle=dotted](\cradius;\xennarg|0,\lambdamsp)(\lambdamin;90)
%
%
\pswedge[linestyle=none](0,0){\cradius}{\xoneopp}{\pheebarl}
\begin{psclip}{\pspolygon[linestyle=none,origin={\cradius;\xoneopp}]%
(\cradius;\pheebarl)(\lambdamin;90)(\cradiusl;90)}
\pswedge[fillstyle=hlines](\cradius;\xonearg){\cradius}{90}{\pheebarl}
\end{psclip}
\psline[origin={\cradius;\xoneopp}](\cradius;\pheebarl)(\lambdamin;90)
\psline[linestyle=dotted](\cradius;\xonearg)%
(\cradius;\xonearg|0,\lambdamsp)
\psline[linestyle=dotted](\cradius;\xonearg|0,\lambdamsp)(\lambdamin;90)
%
%
\rput{0}(\cradius; 37){$\transl_1(P^-)$}
\rput{0}(\cradius;143){$\transl_n(P^+)$}
\rput{0}( 1.5,12.4){$T$}
\rput{0}( 0.5, 7.6){$P=P^+\cup P^-$}
\rput{0}(12.0, 6.4){$Q$}
\rput{0}(12.0,-6.4){$R$}
%
%
\psarc[](0,0){6.0}{0}{\xennopp}
\psarc[](0,0){6.0}{\xoneopp}{180}
\psarc[](0,0){6.3}{180}{\xennarg}
\psarc[](0,0){6.3}{\xonearg}{360}
\FPset{\tmp}{\pheebytwo}
\FPround\tmp\tmp6
\rput{0}(5.0;\tmp){$\fee$}
\FPneg{\tmp}{\pheebytwo}
\FPround\tmp\tmp6
\rput{0}(5.0;\tmp){$\fee$}
\FPadd{\tmp}{180}{\pheebytwo}
\FPround\tmp\tmp6
\rput{0}(5.0;\tmp){$\fee$}
\FPsub{\tmp}{180}{\pheebytwo}
\FPround\tmp\tmp6
\rput{0}(5.0;\tmp){$\fee$}
\psarc[](0,0){6.3}{\pheebarr}{90}
\psarc[](0,0){6.0}{90}{\pheebarl}
\FPsub{\tmp}{90}{\pheebarr}
\FPdiv{\tmp}{\tmp}{2}
\FPsub{\tmp}{90}{\tmp}
\FPround\tmp\tmp6
\rput{0}(5.0;\tmp){$\feebar$}
\FPsub{\tmp}{\pheebarl}{90}
\FPdiv{\tmp}{\tmp}{2}
\FPadd{\tmp}{90}{\tmp}
\FPround\tmp\tmp6
\rput{0}(5.0;\tmp){$\feebar$}
\end{pspicture}
}
\caption{\label{fig:proof2}
An arbitrary set $\sxs\subseteq S^1\subseteq\R^2$ given
in base gap parametrization. Only $\xx_1$ and $\xx_n$ are displayed,
the remaining elements of $\sxs$ are above $\xx_1$ and $\xx_n$.
Recall that $\fee+\feebar=\frac\pi6$.
$R$ is the open set bounded from above by the lower dashed
lines. $Q$ is the closed set between the dashed lines. The
set $P$ is given by the central hatched area. For small values
of $\fee$,
$\transl_1 (P^-) \setminus(Q\cup R)$ and
$\transl_n (P^+) \setminus(Q\cup R)$ are nonempty.
}
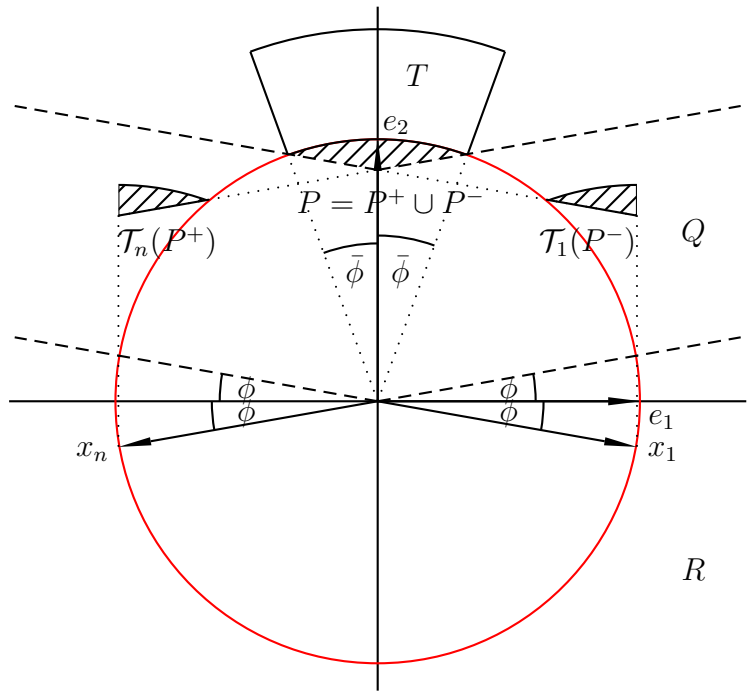
\vspace*{1.5ex}

\dimtwopicture{0.35cm}{10.0}
\end{figure}
\end{center}

\begin{lemma}\label{lemma:proof2helper}
Let $\sxs$ be a finite subset of $S^1\subseteq\R^2$, parametrized
as above. Suppose that $\fee\in[0,\frac{\pi}{6})$, i.e.\ the
size of the base gap is greater than $\frac23\pi$. 
Define the set $\uset$ by
\[
  \uset := P \,\cup\, \transl_n(P^+) \,\cup\, \transl_1(P^-)
  \,\cup\, Q \,\cup\, R,
\]
where $P^+$, $P^-$ denote the elements of $P$ with
non-negative and non-positive $\bv{1}$-coordinate, respectively.
Then $\uu_j\in\uset$ for all $j\in\NZ$.
\end{lemma}

\begin{proof}
Clearly $\uu_0\in\uset$.
By induction, assume that $\uu_j\in\uset$ for some $j\in\N$.
The proof is complete if all of the following claims are shown to be
true. 
\begin{enumerate}[(a)]
\item If $\uu_j\in Q$, then $\uu_{j+1}\in Q\cup R$.
\item If $\uu_j\in P$, then
$\uu_{j+1}\in\transl_n(P^+) \,\cup\, \transl_1(P^-)$.
\item If $\uu_j\in R$, then
$\uu_{j+1}\in P\cup Q\cup R$.
\item If $\uu_j\in \transl_n(P^+)$, then
$\uu_{j+1}\in P\cup Q\cup R$.
\item If $\uu_j\in \transl_1(P^-)$, then
$\uu_{j+1}\in P\cup Q\cup R$.
\end{enumerate}
If $\uu_j\in P\cup Q$, then $\xx_1$ or $\xx_n$ is chosen
in the next step of the iteration,
i.e.\ $\cc_j \in \{\xx_1,\xx_n\}$. Therefore, (b) is trivial.
Also (a) is true since $\transl_1(Q)$ and $\transl_n(Q)$
have no parts above $Q$. If (d) is true then (e) holds
by symmetry. Hence it suffices to show (c) and (d).

\textbf{Claim (c).} Suppose that $\uu_j\in R$ is arbitrarily fixed.
If $\uw_j\in(\pi+\fee,2\pi-\fee)$, then
from Figure~\ref{fig:proof2} it is clear that translation
of the part of $R$ with such argument $\uw_j$
by an arbitrary unit vector stays inside $P\cup Q\cup R$.

Otherwise, $\uw_j\in[-\fee,\fee)$ or $\uw_j\in(\pi-\fee,\pi+\fee]$,
where the second part follows from the first by symmetry.
Restricting to $\uw:=\uw_j\in[-\fee,\fee)$ and setting
$\la:=\la_j>0$, $\cw:=\cw_j\in[\pi+2\uw-\fee,\pi+\fee]$ we
can write
\[
  \uu_{j+1} =
    (\la\cos\uw+\cos\cw)\bv{1} +
    (\la\sin\uw+\sin\cw)\bv{2}.
\]
The range of $\cw$ follows since the center of the interval
of possible values for $\cw$ is $\uw+\pi$, it extends by
$\pi+\fee-(\uw+\pi)=\fee-\uw$ to both sides.
We continue to work on two cases.
\begin{enumerate}[(c.i)]
\item The $\bv{1}$-coordinate of $\uu_{j+1}$ is non-negative.
In this case $\sin(\cw-\fee)\le\frac{\sqrt3}{2}$ and
$\la\sin(\fee-\uw)\ge0$. Since equality does not hold
simultaneously,
\[
  0 < \la\sin(\fee-\uw) + \sin(\fee-\cw) + \frac{\sqrt3}{2}.
\]
Expanding and rearranging the trigonometric terms, substituting
$\la_{min} = \frac{\sqrt3}{2\cos\fee}$ (which denotes the
length of the intersection of $Q$ with the $\bv{2}$-axis) and
dividing by $\cos\fee>0$ we get
\[
  (\la\sin\uw+\sin\cw) - \la_{min} <
    \tan\fee\,(\la\cos\uw+\cos\cw).
\]
This shows that $\uu_{j+1}$ falls below the line bounding $Q$
from above. Hence $\uu_{j+1}\in Q\cup R$.
\item The $\bv{1}$-coordinate of $\uu_{j+1}$ is negative,
i.e.\ $\la<-\frac{\cos\cw}{\cos\uw}$.
If we knew the inequality
\begin{equation}\label{ineq:koppleer}
  \frac{\cos\cw}{\cos\uw} \ge 2\cos(\cw-\uw),
\end{equation}
then $\la\le-2\cos(\cw-\uw)$ would follow
using the inequality for $\la$.
We would arrive at
\[
  \normsqr{\uu_{j+1}} = 1 + \la^2 + 2\la\cos(\cw-\uw) \le 1,
\]
which would show that $\uu_{j+1}\in P\,\cup\,Q\,\cup\,R$.
Hence we are left with \eqref{ineq:koppleer}.
First consider the case $\uw\ge0$.
Then $2\cos(\cw-\uw)<-\sqrt{3}$ and
\[
  \frac{\cos\cw}{\cos\uw} \ge -\frac{1}{\cos\uw} >
  -\frac{2}{\sqrt3},
\]
hence \eqref{ineq:koppleer} is true for this case.
Now restrict to the
case when $\uw<0$. Then $2\cos(\cw-\uw)<-1$ and
\[
  \frac{\cos\cw}{\cos\uw} \ge
  -\frac{\cos(\pi+2\uw-\fee)}{\cos\uw} > -1,
\]
hence \eqref{ineq:koppleer} is true.
\end{enumerate}

\textbf{Claim (d).} From the assumption there is some
$v=\polar{\la}{\delta}\in P^+$ with
$\frac{\sqrt3}{2\sin(\delta-\fee)}\le\la\le1$ and
$\delta\in[\frac\pi2 - \feebar,\frac\pi2]$ such that
\[
  \uu_j = \transl_n v =
    (\la\cos\delta - \cos\fee) \bv{1} +
    (\la\sin\delta - \sin\fee) \bv{2}.
\]
We are done if we show that $\xx_1$ is chosen for the
next step of the iteration, i.e.\ $\cc_j=\xx_1$.
In this case
\[
  \uu_{j+1} = \la\cos\delta \bv{1} +
     (\la\sin\delta-2\sin\fee) \bv{2}.
\]
$\uu_{j+1}$ has a smaller $\bv{2}$-coordinate than the
original point $v\in P^+$, hence
$\uu_{j+1}\in R\cup Q\cup P^+$. We are left with the
mentioned claim and show that the argument angle $\uw_j$
of $\uu_j$ satisfies $\uw_j\le\pi-\fee$.
From
\[
   \la\sin(\fee+\delta) \ge
   \frac{\sqrt3}{2} \frac{\sin(\fee+\delta)}{\sin(\delta-\fee)} \ge
   \frac{\sqrt3}{2} > \sin 2\fee
\]
we get
\[
   (\la\cos\delta-\cos\phi)\sin\phi \ge
   -\cos\fee (\la\sin\delta - \sin\fee).
\]
Since $\la\sin\delta - \sin\fee > 0$ and $\sin\fee\ge0$
division by these terms does not change the type of inequality.
We obtain
\[
   \cot\uw_j =
   \frac{\la\cos\delta-\cos\phi}{\la\sin\delta - \sin\fee} \ge
   -\cot\fee = \cot(\pi-\fee),
\]
which proves the desired fact.
\end{proof}

\begin{lemma}\label{lemma:proof2disjoint}
In the situation of Lemma~\ref{lemma:proof2helper} we have
$\uset\cap \forbiddenset = \emptyset$.
\end{lemma}

\begin{proof}
By construction
$(P \cup Q\cup R)\cap \forbiddenset=\emptyset$. By symmetry it is
therefore enough to show that
$\transl_n(P^+) \cap \forbiddenset = \emptyset$.
As before, let $u=\polar{\la}{\delta}\in P^+$, where
$\delta\in[\frac\pi2 - \feebar,\frac\pi2]$ and
$\frac{\sqrt3}{2\sin(\delta-\fee)}\le \la \le 1$.
Then
\[
  \transl_n u =
     (\la\cos\delta - \cos\fee)\bv{1} +
     (\la\sin\delta - \sin\fee)\bv{2}.
\]
Starting with
\[
  \la\cos(\delta-\feebar) \le \cos(\delta-\feebar) \le
  \frac{\sqrt{3}}{2} \le \cos(\fee-\feebar),
\]
expanding and dividing by $\la\sin\delta-\sin\fee>0$ and
by $\cos\feebar>0$ we get
\[
  \cot\arg\transl_n u =
  \frac{\la\cos\delta - \cos\fee}{\la\sin\delta - \sin\fee}
  \le
  -\tan\feebar = \cot\Bigl(\frac\pi2+\feebar\Bigr),
\]
which shows that the argument angle of $\transl_n u$
is greater or equal than $\frac\pi2+\feebar$.
Therefore $\transl_n u\not\in \forbiddenset$, which proves the assertion.
\end{proof}

\begin{proof}[Proof of Theorem~\ref{thm:main2}.]
Again, the set $A_{2,1}$ from Example~\ref{example:perp} below shows
that $\uastast{2}{1}\ge\sqrt2$.
Moving $\bv{1}$ slightly away from $\bv{2}$ turns $A_{2,1}$ into
a balanced set and shows that also $\uastast{2}{2}\ge\sqrt2$.
Hence it suffices to prove
$\uastast{2}{1},\uastast{2}{2}\le\sqrt2$.
Contrarily, we assume that
there exists an iteration such that $\la_i>\sqrt2$
for some fixed $i\in\N$. Without loss of generality we may assume
that $i$ is the smallest such index, in particular
$\la_{i-1}\le\sqrt2$.

The angle $\gamma_j\in[0,\pi]$ between $\uu_j$ and $\cc_j$ is
defined for all $j\in\N$ since without loss of generality we
may assume $\uu_j\not=0$. Now observe that
\[
  \frac{\pi}{2} + \fee = \frac12(2\pi - (\pi - 2\fee))
    \le\gamma_j\le\pi
\]
for all $j\in\N$.
A simple computation yields
\begin{equation}\label{eqn:lq}
  \la_j^2 = 1 + 2\la_{j-1} \cos\gamma_{j-1} + \la_{j-1}^2.
\end{equation}
Hence 
\[
   2 \la_{i-1}\cos\gamma_{i-1} = \la_i^2 - \la^2_{i-1} - 1 >
   2 - 2 - 1 = -1,
\]
and
\[
  -\frac12 < -\frac{1}{2\la_{i-1}} < \cos\gamma_{i-1} \le
  \cos\Bigl(\frac{\pi}{2} + \fee\Bigr) = -\sin\fee,
\]
since from \eqref{eqn:lq} we also have $1<\la_{i-1}$.
Therefore
\[
  \frac{\pi}{2} + \fee \le \gamma_{i-1} \le \frac23\pi
  \qquad\text{and}\qquad
  0 \le \fee < \frac{\pi}{6}.
\]
In other words there is a gap greater than $\frac23\pi$
between two neighboring elements of $\sxs$. In a second
step of the proof we will explore possible ranges of $\uw_{i-1}$.
Clearly, the angle between $u_{i-1}$ and $\xx_1$, $\xx_n$
is less or equal than $\frac23\pi$. Therefore exactly
one of the following cases holds.

\textbf{Case 1.} $\uw_{i-1}\in(\frac\pi2-\feebar,\frac\pi2+\feebar)$,
where $\feebar:=\frac\pi6-\fee$. Hence
$\uu_{i-1}\in \forbiddenset$ but also $\uu_{i-1}\in \uset$
from Lemma~\ref{lemma:proof2helper}.
This contradicts Lemma~\ref{lemma:proof2disjoint}.

\textbf{Case 2.}
$\uw_{i-1}\in(\frac32\pi-\feetwobar,\frac32\pi+\feetwobar)$,
where $\feetwobar:=\frac\pi6+\fee$. We can restrict the range
of $\uw_{i-1}$ further by adding the above condition not only for
$\xx_1$ and $\xx_n$, but for all elements of $\sxs$. Doing so
we get that
\[
  \left\{\begin{array}{ll}
  \frac23\pi > \uw_{i-1} - \xw_j, &
    \text{if $\pi\ge\uw_{i-1}-\xw_j$, and}\\
  \frac43\pi < \uw_{i-1} - \xw_j, &
    \text{if $\pi<\uw_{i-1}-\xw_j$.}\\
  \end{array} \right.
\]
Let $k=1,\dots,n-1$ be the greatest index satisfying
$\pi < \uw_{i-1}-\xw_k$. Since $k$ is maximal we have
$\pi\ge\uw_{i-1}-\xw_{k+1}$.
We get $\xw_{k+1} - \xw_k > \frac23\pi$, which shows
that there must be a second gap which is greater than $\frac23\pi$.
After a rotation of the coordinate system
and renumbering the elements of $\sxs$
we may apply Lemma~\ref{lemma:proof2disjoint} again and obtain
a contradiction.

The indirect assumption must have been wrong
in Cases~1 and 2, hence both
$\uastast{2}{1},\uastast{2}{2}\le\sqrt2$.
\end{proof}

\section{Examples}

This section provides examples illustrating that the situation
is more complicated in dimension $d\ge3$.
All examples are unique up to rotation of $\R^d$.

\begin{example}\label{example:equidistant}
For $l\ge1$ we describe the operation of choosing $l+1$
equidistant points $\xx_0,\ldots,\xx_l\in S^{l-1}\subseteq\R^l$.
Equidistant means that the value $s$ of the
scalar product does not depend on the chosen pair of points.
Since all vectors have unit length, the constant scalar product
equals $\cos\alpha$ for some $\alpha\in[0,\pi]$. 
By recursion on $l$ suppose $\tilde\xx_1,\dots,\tilde\xx_l$
have been found in the next lower dimension $l-1$,
with scalar product $\tilde s$.
Set
\[
  \xx_0 = (0,0,\ldots,0,1),\quad
  \xx_1 = ( \tilde\xx_1 \cos\alpha, \sin\alpha),\quad
  \dots,\quad
  \xx_l = ( \tilde\xx_l \cos\alpha, \sin\alpha).
\]
We demand
\[
  \sin\alpha = \scp{\xx_0}{\xx_1} = s = \scp{\xx_i}{\xx_j} =
  \sin^2\alpha + \scp{\tilde\xx_i}{\tilde\xx_j} \cos^2\alpha,
\]
which leads to $s=s^2+(1-s^2)\tilde s$.
Solving this equation gives $s=\frac{\tilde s}{1-\tilde s}$.
It is easy to see that the recursion produces the values
\[
  -1, -\frac12, -\frac13, -\frac14, \ldots
\]
for $s$. Hence, when denoting the scalar product of dimension $l$
by $s_l$, we get $s_l=s=-\frac{1}{l}$.
Knowing $s$ it is also clear that $\xx_0+\ldots+\xx_d=0$ since
$\tilde\xx_1 + \ldots + \tilde\xx_d=0$.
In low dimensions, equidistant points are just two points
on the real line ($l=1$), a regular triangle in a circle ($l=2$),
or a tetrahedron in a 2-sphere ($l=3$).
\end{example}

Clearly, the set $\sxs$ of $d+1$ equidistant points is balanced in
$S^{d-1}\subseteq\R^d$. The problem of finding $\uast(\sxs)$ in this case
was approached by a computer experiment only.
We checked $d=2,\ldots,12$ and found that
$\uast(\sxs) = \frac{a(d)}{d}$, where $a$ is the integer
sequence
\[
   0, 1, 2, 4, 6, 9, 12, 16, 20, 25, 30, 36, 42,\ldots
\]
starting at index $d=0$. Obviously, $\uu_i$ may take only
a certain finite number of values on the lattice
\[
  \Bigl\{\,
    \sum_{i=1}^{d+1} k_i \xx_i
    \,\,\bigl|\bigr.\,\,
    k_i\in\NZ
  \,\Bigr\},
\]
all of which are close to the origin.
For example, there are $3$ possibilities for $d=1$ and $7$ for $d=2$.
The sequence $a$ has relations to other fields and problems
\cite{ATT}. Note also that $a(d)<d\sqrt{d}$, or equivalently
$\uast(\sxs)\le\sqrt{d}$. The latter inequality was an ad-hoc conjecture
for a general set $\sxs$, which turned out to be true only in dimension $d=2$.

\begin{example}\label{example:perp}
For $1\le m\le d$
consider the following set $\sxs=A_{d,m}$ consisting of $n=d+m$ points.
As before, let $e_i\in\R^d$ be the vector with all zero components except
the $i$th which is 1. Then define
\[
  A_{d,m} := \{
    e_1,  e_2, \ldots,  e_d,
   -e_1, -e_2, \ldots, -e_m
  \}.
\]
\end{example}

\begin{proposition}\label{proposition:perp}
Let $\sxs=A_{d,m}$ be as in Example~\ref{example:perp}.
\begin{enumerate}
\item $A_{d,m}$ is $m$-balanced,
\item $\uast(A_{d,m}) \ge \sqrt{d-m+1}$.
\end{enumerate}
\end{proposition}
\begin{proof}
(i) is clear from the definition; the origin is contained
in the $m$-dimensional face of $\conv{A_{d,m}}$
spanned by $\pm e_1,\dots,\pm e_m$.
For (ii) observe that there is an iteration such that
$\uu_i=e_{m+1}+e_{m+2}+\ldots+e_{m+i}$ for $1\le i\le d-m$.
\end{proof}

It is likely that equality holds in (ii),
but we do not need this stronger assertion.

\begin{example}\label{example:unbounded}
The following construction of $\sxs=B_{d,b}(\eps,\phi)$
depends on the dimension $d$, some integer $1\le b\le d-2$,
some real numbers $\eps>0$ and $0<\phi<\frac{\pi}{2}$, where
the value of $\phi$ is uncritical.
For $c:=d-b$, $2\le c\le d-1$, we have
the orthogonal decomposition $\R^d = \R^{b} \oplus \R^{c}$.
The subspaces contain unit hyperspheres
$S^{b-1}\subseteq \R^{b}$ and
$S^{c-1}\subseteq \R^{c}$.

In $S^{c-1}$ choose $c+1$ points $\xx_0,\xx_1,\dots,\xx_c$
as follows. Fix any direction $v\in S^{c-1}$ and consider
the linear hyperplane $V$ which is perpendicular to $v$.
In $S^{c-2}=V\cap S^{c-1}$ choose $c$ equidistant points
$\bar\xx_1,\dots,\bar\xx_c$ as described in
Example~\ref{example:equidistant}.
Then let
\[
   \xx_i := \cos(\eps)\,\bar\xx_i + \sin(\eps)\,v
\]
for $i=1,\dots,c$. Note that $\xx_1,\dots,\xx_c$ are equidistant
in $S^{c-2}(\cos\eps) := (V+\sin(\eps)v) \cap S^{c-1}$.
The remaining point $\xx_0$ is given by
\[
   \xx_0 := -\cos(\phi)\,\xx_1 + \sin(\phi)\,v.
\]

In $S^{b-1}$ choose $b+1$ equidistant points $\xx_{c+1},\dots,\xx_{d+1}$,
which makes a total of $n=d+2$ points in $\sxs$.
\end{example}


\begin{proposition}\label{proposition:unbounded}
For $d\ge3$ and $\sxs=B_{d,b}(\eps,\phi)$ the following statements are true.
\begin{enumerate}
\item $\sxs$ is $b$-balanced,
\item for any large $M>0$ there is an $\eps>0$ such
that $\uast(\sxs)\ge\sqrt{M}$.
\end{enumerate}
\end{proposition}
\begin{proof}
(i) is clear from the definition; the origin is contained
in the $b$-dimensional face spanned by $\xx_{c+1},\dots,\xx_{d+1}$.
Note that $\xx_1+\ldots+\xx_c = c\sin(\eps)v$ and
\[
  \sigma := \scp{\xx_i}{\xx_j} = \scp{\bar\xx_i}{\bar\xx_j}
    \cos^2\eps + \sin^2\eps = 1 - \frac{c}{c-1} \cos^2\eps
\]
since $\scp{\bar\xx_i}{\bar\xx_j} = -\frac{1}{c-1}$ for all
$1\le i,j\le c$.
From now on we suppose that $\eps$ is sufficiently small such that
\begin{equation}\label{ineq:sigma}
-\frac{1}{c-1} < \sigma < 0.
\end{equation}
We also have
\[
  \scp{\xx_0}{\xx_i} = \left\{
  \begin{array}{lll}
  -        \cos\phi &+ \sin\phi \sin\eps;& i=1,\\
  - \sigma \cos\phi &+ \sin\phi \sin\eps;& 1<i\le c.
  \end{array}
  \right.
\]
To prove (ii), we show that the iteration which starts
with $\xx_0$ and adds points from $\{\xx_1,\ldots,\xx_{c}\}$
as long as possible is feasible. More precisely,
\[
  \uu_0 = 0,\qquad
  \uu_1 = \xx_0,\qquad
  \uu_2 = \xx_0 + \xx_1,\quad\dots,\quad
  \uu_{c+1} = \xx_0 + \xx_1 + \dots + \xx_c.
\]
In general for $i=0,1,\dots$ we can write
\begin{align}
  \uu_{ic+1} &= \xx_0 + (i-1)(\xx_1 + \xx_2 + \ldots + \xx_c),\nonumber\\
  \uu_{ic+2} &= \xx_0 + (i-1)(\xx_1 + \xx_2 + \ldots + \xx_c) + \xx_1,\nonumber\\
  &\vdots\nonumber\\
  \uu_{ic+c} &= \xx_0 + (i-1)(\xx_1 + \xx_2 + \ldots + \xx_c) +
    (\xx_1 + \xx_2 + \ldots + \xx_{c-1}),\nonumber\\
  \uu_{(i+1)c+1} &= \xx_0 + i(\xx_1 + \xx_2 + \ldots +
  \xx_c).\label{eqniter}
\end{align}
In what follows we fix $0\le i\le k$ and $0\le j\le c-1$ arbitrarily,
and consider step $s:=(i+1)c+j+1$ of the iteration \eqref{eqniter}.
In other words, we want to control the iteration up to and including
step $(k+1)c+m+1$, where $0\le m\le c-1$.
\begin{enumerate}[(a)]
\item To be able to choose $\xx_{j+1}$ in step $s$ we must have
\[
  \scp{\uu_{s}}{\xx_{j+1}} \le 0.
\]
\item Also, to make the choice of $\xx_{j+1}$ work, the scalar product
with all other vectors must be at least as big as the one
from (a), or
\[
  \scp{\uu_{s}}{\xx_{l+1}} \ge
  \scp{\uu_{s}}{\xx_{j+1}}
\]
for all $0\le l\le c-1$.
\item The point $\xx_0$ must not come into play, which is the case when
\[
  \scp{\uu_{s}}{\xx_0} \ge 0.
\]
\item By construction we have
\[
   \scp{\uu_{s}}{\xx_{r+1}}=0
\]
for $c\le r\le d$.
\end{enumerate}

Let us now analyze these conditions. There is nothing to
show for (d). For (c) we compute
\[
  \scp{\uu_{s}}{\xx_0} = \left\{
  \begin{array}{lll}
  1 + ic\sin\eps\sin\phi;& j=0,\\
  1 - \cos\phi + ic\sin\eps\sin\phi - (j-1)\sigma\cos\phi +
    j\sin\eps\sin\phi;& 0<j\le c-1.\\
  \end{array}
  \right.  
\]
From this expression it is clear that (c) is always satisfied.
Looking at (a) and (b) and
observing that $1+(c-1)\sigma = c\sin^2\eps$ we compute
\[
  \scp{\uu_{s}}{\xx_{j+1}} = \left\{
  \begin{array}{lll}
  i\,c\sin^2\eps - \cos\phi &+ \sin\phi \sin\eps;& j=0,\\
  i\,c\sin^2\eps + j\sigma - \sigma \cos\phi &+ \sin\phi \sin\eps;&
    0<j\le c-1
  \end{array}
  \right.
\]
and for $l\not=j$
\[
  \scp{\uu_{s}}{\xx_{l+1}} = \left\{
  \begin{array}{llll}
  ic\sin^2\eps &+ (j-1)\sigma +1-      \cos\phi& + \sin\phi\sin\eps;&
    0=l<j,\\
  ic\sin^2\eps &+ (j-1)\sigma +1-\sigma\cos\phi& + \sin\phi\sin\eps;&
    0<l<j,\\
  ic\sin^2\eps &+ j\sigma -      \sigma\cos\phi& + \sin\phi\sin\eps;&
    l>j.
  \end{array}
  \right.
\]
From these expressions (b) is immediately clear; one just has to
compare the varying terms and to use \eqref{ineq:sigma}.
It remains to analyze Condition (a).
For $j=0$ it can be expressed as
\begin{equation}\label{ineq:a1}
   i \le \frac{\cos\phi - \sin\phi\sin\eps}{c\sin^2\eps},
\end{equation}
for $j>0$ note that we have a set of $c-1$ inequalities, whose
``sharpness'' increases with $j$, cf.~\eqref{ineq:sigma}.
Therefore it suffices to take
the last condition ($j=c-1$) which reads
\begin{equation}\label{ineq:a2}
   i \le \frac{\sigma(\cos\phi-(c-1)) - \sin\phi\sin\eps}{c\sin^2\eps}.
\end{equation}

In the second and last part of the proof,
the assertion is brought into play.
Assume the length $\sqrt{M}$ is reached in step $(k+1)c+m+1$, i.e.
\begin{equation}\label{ineq:length}
  \normsqr{\uu_{(k+1)c+m+1}} \ge M.
\end{equation}
For arbitrary $k$ and $1\le m\le c-1$ we have
\begin{eqnarray*}
\normsqr{\uu_{(k+1)c+m+1}} &=&
    1 + (kc + 2m)kc \sin^2\eps +
    \bigl(1+(m-1)\sigma\bigr)(m-2\cos\phi) +\\
 && 2(kc + m)\sin\eps \sin\phi,
\end{eqnarray*}
while for $m=0$ we get the simpler expression
\begin{equation}\label{ineq:lengthsimple}
  \normsqr{\uu_{(k+1)c + 1}} = 1 + k^2 c^2 \sin^2\eps +
    2kc\sin\eps \sin\phi.
\end{equation}
Assuming $m=0$ (to use the advantages of the simpler form)
and inserting \eqref{ineq:lengthsimple} into \eqref{ineq:length} 
we get an inequality which is quadratic in $k$:
\[
   k^2 + k \frac{2}{c} \frac{\sin\phi}{\sin\eps} +
     \frac{1-M}{c^2\sin^2\eps} \ge 0.
\]
Solving the inequality gives
\begin{equation}\label{ineq:k}
   k \ge \frac{\sqrt{\sin^2\phi - 1 + M} - \sin\phi}{c\sin\eps}.
\end{equation}
To finish the proof, we must put together \eqref{ineq:a1}
and \eqref{ineq:k} as well as \eqref{ineq:a2} and \eqref{ineq:k}.
For the first pairing, solve
\[
   \sqrt{\sin^2\phi - 1 + M} - \sin\phi \le
   \frac{\cos\phi - \sin\phi\sin\eps}{\sin\eps}.
\]
Isolating $M$ yields
\[
   M \le \cos^2\phi \left(1 + \frac{1}{\sin^2\eps}\right).
\]
For small $\eps$, the right-hand side becomes arbitrarily
large, which finishes this part of the proof.
For the remaining pairing, one has to solve
\[
   \sqrt{\sin^2\phi - 1 + M} - \sin\phi \le
   \frac{\sigma(\cos\phi - (c-1)) - \sin\phi\sin\eps}{\sin\eps}.   
\]
Isolating $M$ again gives
\[
   M \le \frac{\sigma^2 (\cos\phi - (c-1))^2}{\sin^2\eps} + \cos^2\phi,
\]
which with small $\eps$ again has an arbitrarily large
right-hand side.
\end{proof}

\begin{example}\label{example:unbounded2}
The following construction of a point set $\sxs=C_{d}(\eps,\mu,\phi)$
depends on the dimension $d\ge3$,
on real numbers $\eps\ge 0$, $\mu>0$ and $0<\phi<\frac{\pi}{2}$, where
the value of $\phi$ is uncritical.
Pick any unit vector $v\in\R^d$ which determines
a hyperplane $V$ of $\R^d$. In $S^{d-2}\subseteq V$ choose $d$
equidistant points $\bar\xx_1,\dots,\bar\xx_d$
as described in Example~\ref{example:equidistant}.
Then define
\[
   \xx_i := \cos(\eps)\bar\xx_i - \sin(\eps)\,v
\]
for $i=1,\dots,d$. The two remaining points are given by
\begin{align*}
   \xx_{d+1} &= - \cos(\mu)\bar\xx_1 + \sin(\mu)\,v,\\
   \xx_0 &= \cos(\phi)\bar\xx_1 + \sin(\phi)\,v.
\end{align*}
Finally let $\sxs:=\{ \xx_0,\xx_1,\dots,\xx_d,\xx_{d+1} \}$.
\end{example}

\begin{proposition}\label{proposition:unbounded2}
For $d\ge 3$ the following statements are true.
\begin{enumerate}
\item $C_d(\eps,\mu,\phi)$ is $d$-balanced for $\eps>0$,
and $(d-1)$-balanced for $\eps=0$, 
\item for any large $M>0$ there is an $\eps>0$ such
that $\uast(C_d(\eps,3\eps,\frac\pi6))\ge\sqrt{M}$,
\item for any large $M>0$ there is a $\mu>0$ such
that $\uast(C_d(0,\mu,\frac\pi6))\ge\sqrt{M}$.
\end{enumerate}
\end{proposition}

\begin{proof}
(i) is immediately clear from the definition, in particular for
$\eps=0$ the origin is contained in the $(d-1)$-dimensional face
spanned by $\xx_1,\dots,\xx_d$. We are left with (ii) and (iii)
which are shown simultaneously.
Consider the following finite piece of an iteration for
$C_d(\eps,\mu,\phi)$. Start with $\uu_0=0$, and let
\begin{align*}
  \uu_1 &= \xx_0,\\
  \uu_2 &= \xx_0 + \xx_{d+1},\\
  \uu_3 &= \xx_0 + \xx_1 + \xx_{d+1},\\
  &\vdots\\
  \uu_{2k-1} &= \xx_0 + (k-1)(\xx_1 + \xx_{d+1}),\\
  \uu_{2k}   &= \xx_0 + (k-1)(\xx_1 + \xx_{d+1}) + \xx_{d+1},\\
  \uu_{2k+1} &= \xx_0 + k(\xx_1 + \xx_{d+1}).
\end{align*}
The following conditions (a)--(c) are sufficient
for the iteration to work as above, up to step $2k+1$.
\begin{enumerate}[(a)]
\item We must have $\scp{\uu_l}{\xx_0}\ge 0$ for all $1\le l\le 2k+1$,
i.e.\ $\xx_0$ is never chosen between steps $2$ and $2k+1$ of the
iteration.
\item Additionally, also the scalar product
with the other vector must be at least as big as the chosen one,
meaning
\[
  \scp{\uu_{2i}}{\xx_1} \le \scp{\uu_{2i}}{\xx_{d+1}},
  \quad
  \scp{\uu_{2i+1}}{\xx_{d+1}} \le \scp{\uu_{2i+1}}{\xx_1}
\]
for all $1\le i\le k$.
\item To be able to choose $\xx_{d+1}$ in step $2i$ and $\xx_1$
in step $2i+1$ we must have
\[
  \scp{\uu_{2i}}{\xx_1} \le \scp{\uu_{2i}}{\xx_m},\quad
  \scp{\uu_{2i+1}}{\xx_{d+1}} \le \scp{\uu_{2i+1}}{\xx_m},
\]
for all $1\le i\le k$ and $2\le m\le d$.
\end{enumerate}

In order to examine Condition (a) it is straightforward to compute
\[
  \scp{\uu_l}{\xx_0} = \left\{
  \begin{array}{lll}
  1 - \cos(\phi+\eps)
  &+\,\, i \bigl(\cos(\phi+\eps) - \cos(\phi+\mu)\bigr);& l=2i,\\
  1
  &+\,\, i \bigl(\cos(\phi+\eps) - \cos(\phi+\mu)\bigr);& l=2i+1.
  \end{array}
  \right.
\]
Since $\mu>\eps$ for both (ii) and (iii), the terms on the right-hand
side are always non-negative. Therefore (a) does not impose any
additional condition. Similarly, for Condition (b) we compute
\[
  \scp{\uu_l}{\xx_1} = \left\{
  \begin{array}{lll}
  \cos(\phi+\eps) - 1
  &+\,\, i \bigl(1 - \cos(\mu - \eps)\bigr);& l=2i,\\
  \cos(\phi+\eps)
  &+\,\, i \bigl(1 - \cos(\mu - \eps)\bigr);& l=2i+1,
  \end{array}
  \right.
\]
\[
  \scp{\uu_l}{\xx_{d+1}} = \left\{
  \begin{array}{lll}
  \cos(\mu-\eps) - \cos(\phi+\mu)
  &+\,\, i \bigl(1 - \cos(\mu - \eps)\bigr);& l=2i,\\
  -\cos(\phi+\mu)
  &+\,\, i \bigl(1 - \cos(\mu - \eps)\bigr);& l=2i+1,
  \end{array}
  \right.
\]
which is equivalent to
\begin{align*}
\cos(\phi+\eps) - 1 &\le \cos(\mu-\eps) - \cos(\phi+\mu),\\
-\cos(\phi+\mu)     &\le \cos(\phi+\eps).
\end{align*}
Again, since both inequalities are always true, (b)
does not introduce new conditions either. 
Finally, Condition (c) requires
\begin{align*}
  \scp{\uu_{2i}}{x_m} - \scp{\uu_{2i}}{\xx_1} &=
    \frac{d}{d-1} \cos\eps \bigl(
      \cos\eps - \cos\phi + i (\cos\mu - \cos\eps)
    \bigr)
  \,\,\ge\,\, 0,\\
  \scp{\uu_{2i+1}}{x_m} - \scp{\uu_{2i+1}}{\xx_{d+1}} &=
    -\frac{d}{d-1} \cos\phi\cos\eps + \cos(\phi+\eps) +
    \cos(\phi+\mu) +\\
     & \qquad i\,\frac{d}{d-1} (\cos\mu - \cos\eps)\cos\eps
  \,\,\ge\,\, 0.
\end{align*}
We demand that if $i$ satisfies the first inequality,
then it shall also satisfy the second. This leads to
the additional condition
\[
  \frac{\cos\phi - \cos\eps}{\cos\mu - \cos\eps} \le
  \frac{\frac{d}{d-1}\cos\phi\cos\eps - \cos(\phi+\eps) -
  \cos(\phi+\mu)}{\frac{d}{d-1}(\cos\mu - \cos\eps)\cos\eps},
\]
which is satisfied if $\frac34\le\cos\phi$, which is the
reason for the choice of $\phi=\frac\pi6$.
Summing up we are left with the condition
\begin{equation}\label{ineq:cai}
  i \le \frac{\cos\phi - \cos\eps}{\cos\mu - \cos\eps}.
\end{equation}

We can now finish the proof for (ii) and (iii).
If the length $\sqrt{M}$ is reached in step $2k+1$, then
we have
\[
  \normsqr{\uu_{2k+1}} = 1 +
  2 k   \bigl(\cos(\phi+\eps) - \cos(\phi+\mu)\bigr) +
  2 k^2 \bigl(1-\cos(\mu-\eps)\bigr) \ge M.
\]
Solving the quadratic inequality in $k$ and using standard
trigonometric identities we get
\begin{equation}\label{ineq:clc}
  k \ge \frac{\sqrt{
  \sin^2(\phi + \frac{\mu+\eps}{2}) + M - 1} -
  \sin  (\phi + \frac{\mu+\eps}{2})}{2 \sin\frac{\mu-\eps}{2}}.
\end{equation}
Putting together \eqref{ineq:cai} and \eqref{ineq:clc} we get
\[
  \frac{\cos\phi - \cos\eps}{\cos\mu - \cos\eps} \ge
  \frac{\sqrt{
  \sin^2(\phi + \frac{\mu+\eps}{2}) + M - 1} -
  \sin  (\phi + \frac{\mu+\eps}{2})}{2 \sin\frac{\mu-\eps}{2}}.
\]
Finally we isolate $M$ and arrive at
\[
  M \le \frac{(\cos\eps-\cos\phi)^2}{\sin^2\frac{\mu+\eps}{2}} +
  \frac{2(\cos\eps-\cos\phi)\sin(\phi+\frac{\mu+\eps}{2})}{
  \sin\frac{\mu+\eps}{2}} + 1.
\]
For (ii) replace $\mu$ by $3\eps$, for (iii) set $\eps=0$.
In both cases the right-hand side becomes arbitrarily large
when $\eps$ resp.\ $\mu$ approaches zero.
\end{proof}

\bibliography{notes}
\bibliographystyle{amsalpha}
\end{document}